\documentclass[AMA,STIX1COL]{WileyNJD-v2}
\usepackage{moreverb}

\def\cl{{C}\!\ell}

\def\P{{\rm P}}

\def\mod{{\rm \;mod\; }}

\newcommand{\BR}{\mathbb{R}}
\newcommand{\BC}{\mathbb{C}}
\newcommand{\Mat}{{\rm Mat}}
\newcommand{\Det}{{\rm Det}}
\newcommand{\tr}{{\rm tr}}
\newcommand{\rank}{{\rm rank}}

\newcommand{\diag}{{\rm diag}}
\newcommand{\Adj}{{\rm Adj}}
\def\cl{\mathcal {G}}
\newcommand{\U}{{\rm U}}

\newcommand{\T}{{\rm T}}
\newtheorem{example}{Example}

\newcommand\BibTeX{{\rmfamily B\kern-.05em \textsc{i\kern-.025em b}\kern-.08em
T\kern-.1667em\lower.7ex\hbox{E}\kern-.125emX}}

\articletype{Research article}

\received{<day> <Month>, <year>}
\revised{<day> <Month>, <year>}
\accepted{<day> <Month>, <year>}

\begin{document}

\title{On Rank of Multivectors in Geometric Algebras\protect\thanks{The article was prepared within the framework of the project “Mirror Laboratories” HSE University “Quaternions, geometric algebras and applications”.}}

\author[1,2]{Dmitry Shirokov*}

\authormark{DMITRY SHIROKOV}

\address[1]{ \orgname{HSE University}, \orgaddress{\state{Moscow}, \country{Russia}}}

\address[2]{ \orgname{Institute for Information Transmission Problems of Russian Academy of Sciences}, \orgaddress{\state{Moscow}, \country{Russia}}}

\corres{Dmitry Shirokov. \email{dm.shirokov@gmail.com}}

\presentaddress{HSE University, 101000, Moscow, Russia}

\abstract[Abstract]{We introduce the notion of rank of multivector in Clifford geometric algebras of arbitrary dimension without using the corresponding matrix representations and using only geometric algebra operations. We use the concepts of characteristic polynomial in geometric algebras and the method of SVD. The results can be used in various applications of geometric algebras in computer science, engineering, and physics.}


\jnlcitation{\cname{%
\author{D. Shirokov}} 
(\cyear{2024}), 
\ctitle{On Rank of Multivectors in Geometric Algebras}
}

\maketitle

\section{Introduction}

The notion of rank of matrix is one of the most important concepts of the matrix theory, which is used in different applications -- data analysis, physics, engineering, control theory, computer sciences, etc.

The Clifford geometric algebras can be regarded as unified language of mathematics \cite{ABS, Porteous, Helm}, physics \cite{Hestenes, Doran, BT, Snygg}, engineering \cite{Bayro2}, and computer science \cite{Dorst, Bayro1,hitzer2}. The Clifford geometric algebras are isomorphic to the classical matrix algebras. In particular, the complexified Clifford geometric algebras $\cl^\BC_{p,q}:=\BC\otimes \cl_{p,q}$ are isomorphic to the following complex matrix algebras:
\begin{eqnarray}
\cl^\BC_{p,q}\simeq
\begin{cases}
    \Mat(2^{\frac{n}{2}}, \BC), &\mbox{if $n$ is even,}\\
    \Mat(2^{\frac{n-1}{2}}, \BC)\oplus\Mat(2^{\frac{n-1}{2}}, \BC), &\mbox{if $n$ is odd.}
\end{cases}
\end{eqnarray}
An arbitrary element $M\in\cl^\BC_{p,q}$ (a multivector) can be represented as a complex matrix of the corresponding size
$$N:=2^{[\frac{n+1}{2}]},$$
where square brackets mean taking the integer part. In the case of odd $n$, we deal with block-diagonal matrices with two nonzero blocks of the same size $2^{\frac{n-1}{2}}$.

In this regard, the problem arises of determining the rank of multivectors $M\in\cl^\BC_{p,q}$ without using the matrix representation and using only the operations in Clifford geometric algebras. In this paper, we solve this problem in the case of any dimension. To do this, we use our previous results on SVD and characteristic polynomial in Clifford geometric algebras. Theorems \ref{thrankpr}, \ref{thrankpr2}, \ref{thrank}, \ref{thrankherm} are new. New explicit formulas (\ref{exp1}), (\ref{exp2}) for the cases of dimensions $3$ and $4$ can be used in various applications of geometric algebras in physics, engineering, and computer science.

The paper is organized as follows. In Section \ref{secGA}, we discuss real and complexified geometric algebras (GA) and introduce the necessary notation. In Section \ref{secbeta}, we discuss an operation of Hermitian conjugation in GA, introduce a positive scalar product, a norm, unitary spaces, and unitary groups in GA. Also we discuss faithful representations of GA and present an explicit form on one of them. In Section \ref{secSVD}, we discuss singular value decomposition of multivectors in GA. In Section \ref{secDet}, we discuss a realization of the determinant and other characteristic polynomial coefficients in GA. In Section \ref{secRank}, we introduce the notion of rank of multivector in GA and prove a number of properties of this notion. We prove that this notion does not depend on the choosing of matrix  representation and present another equivalent definition of this notion using only GA operations. Examples for cases of small dimensions are presented. In Section \ref{secRankherm}, we consider the special case of normal multivectors, for which the rank can be determined more simply. The conclusions follow in Section \ref{secConcl}.

\section{Real and Complexified Geometric Algebras}\label{secGA}

Let us consider the real Clifford geometric algebra $\cl_{p,q}$ \cite{Hestenes,Lounesto,Doran,Bulg} with the identity element $e\equiv 1$ and the generators $e_a$, $a=1, 2, \ldots, n$, where $n=p+q\geq 1$. The generators satisfy the conditions
$$
e_a e_b+e_b e_a=2\eta_{ab}e,\qquad \eta=(\eta_{ab})=\diag(\underbrace{1, \ldots , 1}_p, \underbrace{-1, \ldots, -1}_{q}).
$$
Consider the subspaces $\cl^k_{p,q}$ of grades $k=0, 1, \ldots, n$, which elements are linear combinations of the basis elements $e_A=e_{a_1 a_2 \ldots a_k}=e_{a_1}e_{a_2}\cdots e_{a_k}$, $1 \leq a_1<a_2<\cdots< a_k \leq n$, with ordered multi-indices of length $k$. An arbitrary element (multivector) $M\in\cl_{p,q}$ has the form
$$
M=\sum_A m_A e_A\in\cl_{p,q},\qquad m_A\in\BR,
$$
where we have a sum over all multi-indices of length from $0$ to $n$.
The projection of $M$ onto the subspace $\cl^k_{p,q}$ is denoted by $\langle M \rangle_k$.

The grade involution and reversion of a multivector $M\in\cl_{p,q}$ are denoted by 
\begin{eqnarray}
\widehat{M}=\sum_{k=0}^n(-1)^{k}\langle M \rangle_k,\qquad 
\widetilde{M}=\sum_{k=0}^n (-1)^{\frac{k(k-1)}{2}} \langle M \rangle_k.
\end{eqnarray}
We have
\begin{eqnarray}
\widehat{M_1 M_2}=\widehat{M_1} \widehat{M_2},\qquad \widetilde{M_1 M_2}=\widetilde{M_2} \widetilde{M_1},\qquad \forall M_1, M_2\in\cl_{p,q}.\label{invol}
\end{eqnarray}

Let us consider the complexified Clifford geometric algebra $\cl_{p,q}^\BC:=\BC\otimes\cl_{p,q}$ \cite{Bulg}.
An arbitrary element of $M\in\cl^\BC_{p,q}$ has the form
$$
M=\sum_A m_A e_A\in\cl^\BC_{p,q},\qquad m_A\in\BC.
$$
Note that $\cl^\BC_{p,q}$ has the following basis of $2^{n+1}$ elements:
\begin{eqnarray}
e, ie, e_1, ie_1, e_2, i e_2, \ldots, e_{1\ldots n}, i e_{1\ldots n}.\label{basisC}
\end{eqnarray}

In addition to the grade involution and reversion, we use the operation of complex conjugation, which takes complex conjugation only from the coordinates $m_A$ and does not change the basis elements $e_A$ of $\cl_{p,q}$:
$$
\overline{M}=\sum_A \overline{m}_A e_A\in\cl^\BC_{p,q},\qquad m_A\in\BC,\qquad M\in\cl^\BC_{p,q}.
$$
We have
$$
\overline{M_1 M_2}=\overline{M_1}\,\, \overline{M_2},\qquad \forall M_1, M_2\in\cl^\BC_{p,q}.
$$

\section{Hermitian conjugation and unitary groups in Geometric Algebras}\label{secbeta}

Let us consider an operation of Hermitian conjugation $\dagger$ in $\cl^\BC_{p,q}$  (see \cite{unitary,Bulg}):
\begin{eqnarray}
M^\dagger:=M|_{e_A \to (e_A)^{-1},\,\, m_A \to \overline{m}_A}=\sum_A \overline{m}_A (e_A)^{-1}.\label{herm}
\end{eqnarray}
We have the following two equivalent definitions of this operation:
\begin{eqnarray}
&&M^\dagger=\begin{cases}
e_{1\ldots p} \overline{\widetilde{M}}e_{1\ldots p}^{-1}, & \mbox{if $p$ is odd,}\\
e_{1\ldots p} \overline{\widetilde{\widehat{M}}}e_{1\ldots p}^{-1}, & \mbox{if $p$ is even,}\\
\end{cases}\\
&&M^\dagger=
\begin{cases}
e_{p+1\ldots n} \overline{\widetilde{M}}e_{p+1\ldots n}^{-1}, & \mbox{if $q$ is even,}\\
e_{p+1\ldots n} \overline{\widetilde{\widehat{M}}}e_{p+1\ldots n}^{-1}, & \mbox{if $q$ is odd.}\\
\end{cases}
\end{eqnarray}
The operation\footnote{Compare with the well-known operation $M_1 * M_2:=\langle \widetilde{M_1} M_2 \rangle_0$ in the real geometric algebra $\cl_{p,q}$, which is positive definite only in the case of signature $(p,q)=(n,0)$.} 
$$(M_1, M_2):=\langle M_1^\dagger M_2 \rangle_0$$
is a (positive definite) scalar product with the properties
\begin{eqnarray}
&&(M_1, M_2)=\overline{(M_2, M_1)},\\
&&(M_1+M_2, M_3)=(M_1, M_3)+(M_2, M_3),\quad (M_1, \lambda M_2)=\lambda (M_1, M_2),\\
&&(M, M)\geq 0,\quad (M, M)=0 \Longleftrightarrow M=0.\label{||M||}
\end{eqnarray}
Using this scalar product we introduce an inner product space over the field of complex numbers (a unitary space) in $\cl^\BC_{p,q}$.

We have a norm 
\begin{eqnarray}
\|M\|:=\sqrt{(M,M)}=\sqrt{\langle M^\dagger M \rangle_0}.\label{norm}
\end{eqnarray}

Let us consider the following faithful representation (isomorphism) of the complexified geometric algebra
\begin{eqnarray}
\beta:\cl^\BC_{p,q}\quad \to\quad
\begin{cases}
    \Mat(2^{\frac{n}{2}}, \BC), &\mbox{if $n$ is even,}\\
    \Mat(2^{\frac{n-1}{2}}, \BC)\oplus\Mat(2^{\frac{n-1}{2}}, \BC), &\mbox{if $n$ is odd.}
\end{cases}\label{isom}
\end{eqnarray}
Let us denote the size of the corresponding matrices by
$$N:=2^{[\frac{n+1}{2}]},$$
where square brackets mean taking the integer part.

Let us present an explicit form of one of these representations of $\cl^\BC_{p,q}$ (we use it also for $\cl_{p,q}$ in \cite{det} and for $\cl^\BC_{p,q}$ in \cite{LMA}). We denote this fixed representation by $\beta'$. Let us consider the case $p = n$, $q = 0$. To obtain the matrix representation for another signature with $q\neq 0$, we should multiply matrices $\beta'(e_a)$, $a = p + 1, \ldots, n$ by imaginary unit $i$. For the identity element, we always use the identity matrix $\beta'(e)=I_N$ of the corresponding dimension $N$. We always take $\beta'(e_{a_1 a_2 \ldots a_k}) = \beta' (e_{a_1}) \beta' (e_{a_2}) \cdots \beta'(e_{a_k})$. In the case $n=1$, we take $\beta'(e_1)=\diag(1, -1)$. Suppose we know $\beta'_a:=\beta'(e_a)$, $a = 1, \ldots, n$ for some fixed odd $n = 2k + 1$. Then for $n = 2k + 2$, we take
the same $\beta'(e_a)$, $a = 1, \ldots , 2k + 1$, and 
$$\beta'(e_{2k+2})=\left(
    \begin{array}{cc}
      0 & I_{\frac{N}{2}}  \\
      I_{\frac{N}{2}} & 0 
    \end{array}
  \right).$$
For $n = 2k + 3$, we take
$$\beta'(e_{a})= \left(\begin{array}{cc}
      \beta'_a & 0 \\
      0 & -\beta'_a
    \end{array}
  \right),\qquad a=1, \ldots, 2k+2,$$ 
  and 
  $$\beta'(e_{2k+3})=\left(\begin{array}{cc}
      i^{k+1}\beta'_1\cdots \beta'_{2k+2} & 0 \\
      0 & -i^{k+1}\beta'_1\cdots \beta'_{2k+2} 
    \end{array}
  \right).$$
  This recursive method gives us an explicit form of the matrix representation $\beta'$ for all $n$.

Note that for this matrix representation we have
$$
(\beta'(e_a))^\dagger=\eta_{aa} \beta'(e_a),\qquad a=1, \ldots, n,
$$
where $\dagger$ is the Hermitian transpose of a matrix. Using the linearity, we get that Hermitian conjugation of a matrix is consistent with Hermitian conjugation of the corresponding multivector:
\begin{eqnarray}
\beta'(M^\dagger)=(\beta'(M))^\dagger,\qquad M\in\cl^\BC_{p,q}.\label{sogl}
\end{eqnarray}
Note that the same is not true for an arbitrary matrix representations $\beta$ of the form (\ref{isom}). It is only true for matrix representations $\gamma=T^{-1}\beta' T$ obtained from $\beta'$ using matrices $T$ such that $T^\dagger T= I$.

Let us consider the group
\begin{eqnarray}
\U\cl^\BC_{p,q}=\{M\in \cl^\BC_{p,q}: M^\dagger M=e\},
\end{eqnarray}
which we call a unitary group in $\cl^\BC_{p,q}$. Note that all the basis elements $e_A$ of $\cl_{p,q}$ belong to this group by definition.

Using (\ref{isom}) and (\ref{sogl}), we get the following isomorphisms to the classical matrix unitary groups:
\begin{eqnarray}
\U\cl^\BC_{p,q}\simeq\begin{cases}
    \U(2^{\frac{n}{2}}), &\mbox{if $n$ is even,}\\
    \U(2^{\frac{n-1}{2}})\times\U(2^{\frac{n-1}{2}}), &\mbox{if $n$ is odd,}
\end{cases}\label{isgr}
\end{eqnarray}
where
\begin{eqnarray}
\U(k)=\{A\in\Mat(k, \BC),\quad A^\dagger A=I\}.
\end{eqnarray}

\section{Singular Value Decomposition in Geometric Algebras}\label{secSVD}

The method of singular value decomposition was discovered independently by E. Beltrami in 1873 \cite{Beltrami} and C. Jordan in 1874 \cite{Jordan1,Jordan2}.

We have the following well-known theorem on singular value decomposition of an arbitrary complex matrix \cite{For,Van}. For an arbitrary $A\in\BC^{n\times m}$, there exist matrices $U\in \U(n)$ and $V\in\U(m)$ such that
\begin{eqnarray}
 A=U\Sigma V^\dagger,\label{SVD}
\end{eqnarray}
where
$$
\Sigma=\diag(\lambda_1, \lambda_2, \ldots, \lambda_k),\qquad k=\min(n, m),\qquad \BR\ni\lambda_1, \lambda_2, \ldots, \lambda_k\geq 0.
$$
Note that choosing matrices $U\in \U(n)$ and $V\in\U(m)$, we can always arrange diagonal elements of the matrix $\Sigma$
in decreasing order $\lambda_1\geq \lambda_2 \geq \cdots \geq \lambda_k\geq 0$.

Diagonal elements of the matrix $\Sigma$ are called singular values, they are square roots of eigenvalues of the matrices $A A^\dagger$ or $A^\dagger A$. Columns of the matrices $U$ and $V$ are eigenvectors of the matrices $A A^\dagger$ and $A^\dagger A$ respectively.

\begin{theorem}[SVD in GA]\cite{SVDENGAGE, SVDAACA}\label{th1} For an arbitrary multivector $M\in\cl^\BC_{p,q}$, there exist multivectors $U, V\in \U\cl^\BC_{p,q}$, where
$$
\U\cl^\BC_{p,q}=\{U\in \cl^\BC_{p,q}: U^\dagger U=e\},\qquad U^\dagger:=\sum_A \overline{u}_A (e_A)^{-1},
$$
such that
\begin{eqnarray}
M=U\Sigma V^\dagger,\label{SVDMC}
\end{eqnarray}
where multivector $\Sigma$ belongs to the subspace $K\in\cl^\BC_{p,q}$, which is a real span of a set of $N=2^{[\frac{n+1}{2}]}$ fixed basis elements  (\ref{basisC}) of $\cl^\BC_{p,q}$ including the identity element~$e$.
\end{theorem}

\section{Determinant and other characteristic polynomial coefficients in Geometric Algebras}\label{secDet}

Let us consider the concept of determinant \cite{rudn,acus} and characteristic polynomial \cite{det} in geometric algebra. Explicit formulas for characteristic polynomial coefficients are discussed in \cite{Abd,Abd2}, applications to Sylvester equation are discussed in \cite{Sylv,Sylv2}, the relation with the noncommutative Vieta theorem is discussed in \cite{Vieta1,Vieta2}, applications  to calculation of elementary functions in geometric algebras are discussed in \cite{Acus}.

We can introduce the notion of determinant
$$\Det(M):=\det(\beta(M))\in\BR,\qquad M\in\cl^\BC_{p,q},$$
where $\beta$ is (\ref{isom}), and the notion of characteristic polynomial
\begin{eqnarray}
&&\varphi_M(\lambda):=\Det(\lambda e-M)=\lambda^N-C_{(1)}\lambda^{N-1}-\cdots-C_{(N-1)}\lambda-C_{(N)}\in\cl^0_{p,q}\equiv\BR,\nonumber\\
&&M\in\cl^\BC_{p,q},\quad N=2^{[\frac{n+1}{2}]},\quad C_{(k)}=C_{(k)}(M)\in\cl^0_{p,q}\equiv\BR,\quad k=1, \ldots, N.\label{char}
\end{eqnarray}
 The following method based on the Faddeev--LeVerrier algorithm allows us to recursively obtain basis-free formulas for all the characteristic coefficients $C_{(k)}$, $k=1, \ldots, N$ (\ref{char}):
\begin{eqnarray}
&&M_{(1)}:=M,\qquad M_{(k+1)}=M(M_{(k)}-C_{(k)}),\label{FL0}\\
&&C_{(k)}:=\frac{N}{k}\langle M_{(k)} \rangle_0,\qquad k=1, \ldots, N.
\label{FL}\end{eqnarray}
In particular, we have
\begin{eqnarray}
C_{(1)}=N \langle M \rangle_0=\tr(\beta(M)).
\end{eqnarray}
In this method, we obtain high coefficients from the lowest ones. The determinant is minus the last coefficient
\begin{eqnarray}
\Det(M)=-C_{(N)}=-M_{(N)}=U(C_{(N-1)}-M_{(N-1)})\label{laststep}
\end{eqnarray}
and has the properties (see \cite{rudn,det})
\begin{eqnarray}
&&\Det(M_1 M_2)=\Det(M_1) \Det (M_2),\qquad M_1, M_2\in\cl^\BC_{p,q},\label{detpr}\\
&&\Det(M)=\Det(\widehat{M})=\Det(\widetilde{M})=\Det(\overline{M})=\Det(M^\dagger),\qquad \forall M\in\cl^\BC_{p,q}.\label{detpr2}
\end{eqnarray}
The inverse of a multivector $M\in\cl^\BC_{p,q}$ can be computed as
\begin{eqnarray}
M^{-1}=\frac{\Adj(M)}{\Det(M)}=\frac{C_{(N-1)}-M_{(N-1)}}{\Det(M)},\qquad \Det(M)\neq 0.\label{inv}
\end{eqnarray}

The presented algorithm and formulas (\ref{FL0}), (\ref{FL}), (\ref{inv}) are actively used to calculate inverse in GA \cite{inv1,inv2,inv3}. See also \cite{hitzer1}.

\section{Rank in Geometric Algebras}\label{secRank}

Let us introduce the notion of rank of a multivector $M\in\cl^\BC_{p,q}$:
\begin{eqnarray}
\rank(M):=\rank(\beta(M))\in\{0, 1, \ldots, N\},\label{rank}
\end{eqnarray}
where $\beta$ is (\ref{isom}). 

Below we present another equivalent definition, which does not depend on the matrix representation $\beta$ (Theorem \ref{thrank}). We use the fact that rank is the number of nonzero singular values in the SVD and Vieta formulas.

\begin{lemma}\label{lemmawell} The rank of multivector $\rank(M)$ (\ref{rank}) is well-defined, i.e. it does not depend on the representation $\beta$ (\ref{isom}).
\end{lemma}
\begin{proof} In the case of even $n$, for an
arbitrary representation $\beta$ of type (\ref{isom}), by the Pauli theorem \cite{Pauli}, there exists $T$ such that $\beta(e_a)=T^{-1}\beta'(e_a) T$, where $\beta'$ is fixed matrix representation from Section \ref{secbeta}. We get $\beta(M)=T^{-1}\beta'(M) T$ and $\rank(\beta(M))=\rank(\beta'(M))$.

In the case of odd $n$, for an
arbitrary representation $\beta$ of type (\ref{isom}), by the Pauli theorem \cite{Pauli}, there exists $T$ such that $\beta(e_a)=T^{-1}\beta'(e_a) T$ or $\beta(e_a)=-T^{-1}\beta'(e_a) T$. In the first case, we get  $\rank(\beta(M))=\rank(\beta'(M))$ similarly to the case of even $n$. In the second case, we get $\beta(M)=T^{-1}\beta'(\widehat{M}) T$ and $\rank(\beta(M))=\rank(\beta'(\widehat{M}))$. The equality $\rank(\beta'(\widehat{M}))=\rank(\beta'(M))$ is verified using the explicit form of representation $\beta'$ from Section \ref{secbeta}. Namely, the matrices $\beta'(e_a)=\diag(\beta'_a, -\beta'_a)$, $a=1, \ldots, n$, are block-diagonal matrices with two blocks differing in sign on the main diagonal by construction. Thus the matrix $\beta'(e_{ab})=\beta'(e_a)\beta'(e_b)=\diag(\beta'_a \beta'_b, \beta'_a \beta'_b)$ has two identical blocks. We conclude that the even part of multivector $M$ has the matrix representation $\diag(A, A)$ with two identical blocks, and the odd part of multivector $M$ has the matrix representation $\diag(B, -B)$ with two blocks differing in sign. Finally, we obtain $\rank(\beta'(\widehat{M})=\rank(\diag(A-B, A+B))=\rank(\diag(A+B, A-B))=\rank(\beta'(M))$.  
\end{proof}

\begin{theorem}\label{thrankpr}
We have the following properties of the rank of arbitrary multivectors $M_1, M_2, M_3\in\cl^\BC_{p,q}$:
\begin{eqnarray}
&&\rank(M_1 U)=\rank(U M_1)=\rank (M_1),\qquad \forall \,\,\mbox{invertible}\,\,U\in\cl^\BC_{p,q},\\
&&\rank(M_1 M_2)\leq \min(\rank(M_1), \rank(M_2)),\\
&&\rank(M_1 M_2)+\rank(M_2 M_3)\leq \rank(M_1 M_2 M_3)+\rank(M_2),\\
&&\rank(M_1 )+\rank(M_3)\leq \rank(M_1 M_3)+N.
\end{eqnarray}
\end{theorem}
\begin{proof} These properties are the corollary of the corresponding properties of rank of matrices.
\end{proof}

\begin{theorem}\label{thrankpr2} We have
\begin{eqnarray}
\rank(M)=\rank(\widehat{M})=\rank(\widetilde{M})=\rank(\overline{M})=\rank(M^\dagger)=\rank(M^\dagger M)=\rank(M M^\dagger),\qquad \forall M\in\cl^\BC_{p,q}.
\end{eqnarray}
\end{theorem}
\begin{proof} Let us prove $\rank(M)=\rank(\widehat{M})$. In the case of even $n$, we have $\rank(\widehat{M})=\rank(e_{1\ldots n}M e_{1\ldots n}^{-1})=\rank (M)$. In the case of odd $n$, we have already proved the statement in the proof of Lemma \ref{lemmawell}.

Let us prove $\rank(M)=\rank(\widetilde{M})$. We have the following relation between the reversion (or the superposition of reversion and grade involution) and the transpose (see \cite{nspinors,LMA}):
\begin{eqnarray}
(\beta'(M))^\T=\begin{cases}
\beta'(e_{b_1 \ldots b_k}\widetilde{M}e_{b_1\ldots b_k}^{-1}), & \mbox{if $k$ is odd,}\\
\beta'(e_{b_1 \ldots b_k}\widehat{\widetilde{M}}e_{b_1\ldots b_k}^{-1}), & \mbox{if $k$ is even,}
\end{cases}
\end{eqnarray}
for some fixed basis element $e_{b_1\ldots b_k}$, where $k$ is the number of symmetric matrices among $\beta'(e_a)$, $a=1, \ldots, n$. We get $\rank(M)=\rank(\beta'(M))=\rank((\beta'(M))^\T)=\rank(\widetilde{M})$.

Using (\ref{sogl}), we obtain the other formulas for the Hermitian conjugation and complex conjugation, which is a superposition of Hermitian conjugation and transpose.
\end{proof}

\begin{lemma}\label{lemmaB}
Suppose that a square matrix $A\in\BC^{N\times N}$ is diagonalizable. Then
\begin{eqnarray}
\rank(A)=N \quad &\Longleftrightarrow& \quad  C_{(N)}\neq 0;\\
 \rank(A)=k\in\{1, \ldots, N-1\} \quad &\Longleftrightarrow& \quad  C_{(k)}\neq 0,\quad C_{(j)}=0,\quad j=k+1, \ldots, N;\\
\rank(A)=0 \quad &\Longleftrightarrow& \quad  A=0.
\end{eqnarray}
\end{lemma}
\begin{proof}
    We use Vieta formulas for the eigenvalues $\lambda_1, \lambda_2, \ldots, \lambda_N$:
    \begin{eqnarray}
    C_{(1)}&=&\lambda_1+\cdots+\lambda_N,\\
    C_{(2)}&=&-(\lambda_1 \lambda_2+\lambda_1 \lambda_3+\cdots+\lambda_{N-1}\lambda_N),\\
   && \cdots\\
    C_{(N)}&=&-\lambda_1 \cdots \lambda_N.
    \end{eqnarray}
 To the right, all statements are obvious. To the left, they are proved by contradiction.   
\end{proof}

\begin{lemma}\label{lemmaC}
For an arbitrary multivector $M\in\cl^\BC_{p,q}$, we have
\begin{eqnarray}
    C_{(N)}(M^\dagger M)=0 &\Longleftrightarrow& C_{(N)}(M)=0,\\
    C_{(1)}(M^\dagger M)=0 &\Longleftrightarrow& M=0.
\end{eqnarray}
\end{lemma}
\begin{proof} We have
\begin{eqnarray*}
C_{(N)}(M^\dagger M)&=&-\Det(M^\dagger M)=-\Det(M^\dagger) \Det(M)=-(\Det M)^2=(C_{(N)}(M))^2,\\
C_{(1)}(M^\dagger M)&=&N \langle M^\dagger M \rangle_0=N \|M\|^2,
\end{eqnarray*}
where we use (\ref{detpr}), (\ref{detpr2}), (\ref{norm}), and (\ref{||M||}).
\end{proof}

\begin{theorem}[Rank in GA]\label{thrank} Let us consider an arbitrary multivector $M\in\cl^\BC_{p,q}$ and $T:=M^\dagger M$. We have
\begin{eqnarray}
\rank(M)=\begin{cases}
N,\quad &\mbox{if $C_{(N)}(M)\neq 0$,}\\
N-1,\quad &\mbox{if $C_{(N)}(M)=0$ and $C_{(N-1)}(T)\neq 0$,}\\
N-2\qquad &\mbox{if $C_{(N)}(M)=C_{(N-1)}(T)=0$ and $C_{(N-2)}(T)\neq 0$,}\\
\cdots  &\\
2,\quad  &\mbox{if $C_{(N)}(M)=C_{(N-1)}(T)=\cdots=C_{(3)}(T)=0$ and $C_{(2)}(T)\neq 0$,}\\
1,\quad  &\mbox{if $C_{(N)}(M)=C_{(N-1)}(T)=\cdots=C_{(2)}(T)=0$ and $M\neq 0$,}\\
0,\quad  &\mbox{if $M=0$.}\label{rank22}
\end{cases}
\end{eqnarray}
\end{theorem}
\begin{proof}  We use the fact that the rank of a matrix equals the number of non-zero singular values, which is the same as the number of non-zero diagonal elements of the matrix $\Sigma$ in the singular value decomposition $A=U\Sigma V^\dagger$ (\ref{SVD}): $\rank(A)=\rank(U\Sigma V^\dagger)=\rank(\Sigma)$. The number of non-zero diagonal elements of the matrix $\Sigma$ can be written in terms of zero and non-zero characteristic polynomial coefficients of the matrix $A^\dagger A$ (see Lemma \ref{lemmaB}). Then we use Lemma~\ref{lemmaC}. 
\end{proof}

\begin{example}
For an arbitrary $M\in\cl^\BC_{p,q}$, $n=p+q=1$, we have
\begin{eqnarray}
\rank(M)=\begin{cases}
$2,\quad$ &\mbox{if $M\widehat{M}\neq 0$,}\\
$1,\quad$  &\mbox{if $M\widehat{M}=0$ and $M\neq 0$,}\\
$0,\quad$  &\mbox{if $M=0$.}
\end{cases}\label{exp-1}
\end{eqnarray}
\end{example}

\begin{example}
For an arbitrary $M\in\cl^\BC_{p,q}$, $n=p+q=2$, we have
\begin{eqnarray}
\rank(M)=\begin{cases}
$2,\quad$ &\mbox{if $M\widetilde{\widehat{M}}\neq 0$,}\\
$1,\quad$  &\mbox{if $M\widetilde{\widehat{M}}=0$ and $M\neq 0$,}\\
$0,\quad$  &\mbox{if $M=0$.}
\end{cases}\label{exp0}
\end{eqnarray}
\end{example}

\begin{example}
For an arbitrary $M\in\cl^\BC_{p,q}$, $n=p+q=3$, we have
\begin{eqnarray}
\rank(M)=\begin{cases}
$4,\quad$ &\mbox{if $M\widetilde{\widehat{M}}\widehat{M}\widetilde{M}\neq 0$,}\\
$3,\quad$ &\mbox{if $M\widetilde{\widehat{M}}\widehat{M}\widetilde{M}=0$ and $T\widetilde{\widehat{T}}\widehat{T}+T\widetilde{\widehat{T}}\widetilde{T}+T\widehat{T}\widetilde{T}+\widetilde{\widehat{T}}\widehat{T}\widetilde{T}\neq 0$,}\\
$2,\quad$ &\mbox{if $M\widetilde{\widehat{M}}\widehat{M}\widetilde{M}=T\widetilde{\widehat{T}}\widehat{T}+T\widetilde{\widehat{T}}\widetilde{T}+T\widehat{T}\widetilde{T}+\widetilde{\widehat{T}}\widehat{T}\widetilde{T}=0$ and} \\ &\mbox{$T\widetilde{\widehat{T}}+T\widehat{T}+T\widetilde{T}+\widetilde{\widehat{T}}\widehat{T}+\widetilde{\widehat{T}}\widetilde{T}+\widehat{T}\widetilde{T}\neq 0$,}\\
$1,\quad$  &\mbox{if $M\widetilde{\widehat{M}}\widehat{M}\widetilde{M}=T\widetilde{\widehat{T}}\widehat{T}+T\widetilde{\widehat{T}}\widetilde{T}+T\widehat{T}\widetilde{T}+\widetilde{\widehat{T}}\widehat{T}\widetilde{T}=$}\\
&\mbox{$=T\widetilde{\widehat{T}}+T\widehat{T}+T\widetilde{T}+\widetilde{\widehat{T}}\widehat{T}+\widetilde{\widehat{T}}\widetilde{T}+\widehat{T}\widetilde{T}=0$ and $M\neq 0$,}\\
$0,\quad$  &\mbox{if $M=0$,}
\end{cases}\label{exp1}
\end{eqnarray}
where $T:=M^\dagger M$.
\end{example}

\begin{example} Let us consider the $\bigtriangleup$-operation \cite{det}
\begin{eqnarray}
M^{\bigtriangleup}:=\sum_{k=0}^n (-1)^{\frac{k(k-1)(k-2)(k-3)}{24}}\langle M \rangle_k=\sum_{k=0, 1, 2, 3\mod 8}\langle M \rangle_k-\sum_{k=4, 5, 6, 7\mod 8}\langle M \rangle_k.\label{opconj}
\end{eqnarray}
Note that we have $(M_1 M_2)^\bigtriangleup\neq M_1^\bigtriangleup M_2^\bigtriangleup$ and $(M_1 M_2)^\bigtriangleup\neq M_2^\bigtriangleup M_1^\bigtriangleup $ 
in the general case.

For an arbitrary $M\in\cl^\BC_{p,q}$, $n=p+q=4$, we have
\begin{eqnarray}
\rank(M)=\begin{cases}
$4,$ &\mbox{if $M\widetilde{\widehat{M}}(\widehat{M}\widetilde{M})^\bigtriangleup\neq 0$,}\\
$3,$ &\mbox{if $M\widetilde{\widehat{M}}(\widehat{M}\widetilde{M})^\bigtriangleup=0$ and} \\
&\mbox{$T\widetilde{\widehat{T}}\widehat{T}+T\widetilde{\widehat{T}}\widetilde{T}+T(\widehat{T}\widetilde{T})^\bigtriangleup+\widetilde{\widehat{T}}(\widehat{T}\widetilde{T})^\bigtriangleup\neq 0$,}\\
$2,$ &\mbox{if $M\widetilde{\widehat{M}}(\widehat{M}\widetilde{M})^\bigtriangleup=T\widetilde{\widehat{T}}\widehat{T}+T\widetilde{\widehat{T}}\widetilde{T}+T(\widehat{T}\widetilde{T})^\bigtriangleup+\widetilde{\widehat{T}}(\widehat{T}\widetilde{T})^\bigtriangleup=0$} \\ &\mbox{and $T\widetilde{\widehat{T}}+T\widehat{T}+T\widetilde{T}+\widetilde{\widehat{T}}\widehat{T}+\widetilde{\widehat{T}}\widetilde{T}+(\widehat{T}\widetilde{T})^\bigtriangleup\neq 0$,}\\
$1,$  &\mbox{if $M\widetilde{\widehat{M}}(\widehat{M}\widetilde{M})^\bigtriangleup=T\widetilde{\widehat{T}}\widehat{T}+T\widetilde{\widehat{T}}\widetilde{T}+T(\widehat{T}\widetilde{T})^\bigtriangleup+\widetilde{\widehat{T}}(\widehat{T}\widetilde{T})^\bigtriangleup=$}\\
&\mbox{$=T\widetilde{\widehat{T}}+T\widehat{T}+T\widetilde{T}+\widetilde{\widehat{T}}\widehat{T}+\widetilde{\widehat{T}}\widetilde{T}+(\widehat{T}\widetilde{T})^\bigtriangleup=0$ and $M\neq 0$,}\\
$0,$  &\mbox{if $M=0$,}
\end{cases}\label{exp2}
\end{eqnarray}
where $T:=M^\dagger M$.
\end{example}

You can get an explicit form of formulas from Theorem \ref{thrank} for the cases $n=5$ and $n=6$ using the explicit formulas for the characteristic polynomial coefficients $C_{(k)}$, $k=1, 2, \dots, N$, from \cite{Abd2}. We do not present them here because they are quite cumbersome. These formulas involve only the operations of summation, multiplication, $\widehat{\quad}$, $\widetilde{\quad}$, and $\,\,^\bigtriangleup$.

\section{The Case of Normal Multivectors}\label{secRankherm}

We call \textit{a normal multivector} $M\in\cl^\BC_{p,q}$ a multivector with the property $M^\dagger M=M M^\dagger$, where $\dagger$ is (\ref{herm}).

\textit{Hermitian multivectors} $M^\dagger=M$, \textit{anti-Hermitian multivectors} $M^\dagger=-M$, \textit{unitary multivectors }$M^\dagger M=e$ are the particular cases of normal multivectors. For example, the basis elements $e_A$ of $\cl_{p,q}$ are unitary by the definition. Note that all unitary multivectors have rank equal to $N$.

\begin{theorem}\label{thrankherm} Let us consider a normal ($M^\dagger M=M M^\dagger$) multivector $M\in\cl^\BC_{p,q}$. We have
\begin{eqnarray}
\rank(M)=\begin{cases}
N,\quad &\mbox{if $C_{(N)}(M)\neq 0$,}\\
N-1,\quad &\mbox{if $C_{(N)}(M)=0$ and $C_{(N-1)}(M)\neq 0$,}\\
N-2\qquad &\mbox{if $C_{(N)}(M)=C_{(N-1)}(M)=0$ and $C_{(N-2)}(M)\neq 0$,}\\
\cdots  &\\
2,\quad  &\mbox{if $C_{(N)}(M)=C_{(N-1)}(M)=\cdots=C_{(3)}(M)=0$ and $C_{(2)}(M)\neq 0$,}\\
1,\quad  &\mbox{if $C_{(N)}(M)=C_{(N-1)}(M)=\cdots=C_{(2)}(M)=0$ and $M\neq 0$,}\\
0,\quad  &\mbox{if $M=0$.}\label{rankherm}
\end{cases}
\end{eqnarray}
\end{theorem}
\begin{proof}  We use that a normal matrix is always diagonalizable. 
\end{proof}

\begin{example}
For an arbitrary normal ($M^\dagger M=M M^\dagger$) multivector $M\in\cl^\BC_{p,q}$, $n=p+q=3$, we have
\begin{eqnarray}
\rank(M)=\begin{cases}
$4,\quad$ &\mbox{if $M\widetilde{\widehat{M}}\widehat{M}\widetilde{M}\neq 0$,}\\
$3,\quad$ &\mbox{if $M\widetilde{\widehat{M}}\widehat{M}\widetilde{M}=0$ and $M\widetilde{\widehat{M}}\widehat{M}+M\widetilde{\widehat{M}}\widetilde{M}+M\widehat{M}\widetilde{M}+\widetilde{\widehat{M}}\widehat{M}\widetilde{M}\neq 0$,}\\
$2,\quad$ &\mbox{if $M\widetilde{\widehat{M}}\widehat{M}\widetilde{M}=M\widetilde{\widehat{M}}\widehat{M}+M\widetilde{\widehat{M}}\widetilde{M}+M\widehat{M}\widetilde{M}+\widetilde{\widehat{M}}\widehat{M}\widetilde{M}=0$ and} \\ &\mbox{$M\widetilde{\widehat{M}}+M\widehat{M}+M\widetilde{M}+\widetilde{\widehat{M}}\widehat{M}+\widetilde{\widehat{M}}\widetilde{M}+\widehat{M}\widetilde{M}\neq 0$,}\\
$1,\quad$  &\mbox{if $M\widetilde{\widehat{M}}\widehat{M}\widetilde{M}=M\widetilde{\widehat{M}}\widehat{M}+M\widetilde{\widehat{M}}\widetilde{M}+M\widehat{M}\widetilde{M}+\widetilde{\widehat{M}}\widehat{M}\widetilde{M}=$}\\
&\mbox{$=M\widetilde{\widehat{M}}+M\widehat{M}+M\widetilde{M}+\widetilde{\widehat{M}}\widehat{M}+\widetilde{\widehat{M}}\widetilde{M}+\widehat{M}\widetilde{M}=0$ and $M\neq 0$,}\\
$0,\quad$  &\mbox{if $M=0$.}
\end{cases}\label{exp12}
\end{eqnarray}
\end{example}

\begin{example} 
For an arbitrary normal ($M^\dagger M=M M^\dagger$) multivector $M\in\cl^\BC_{p,q}$, $n=p+q=4$, we have
\begin{eqnarray}
\rank(M)=\begin{cases}
$4,$ &\mbox{if $M\widetilde{\widehat{M}}(\widehat{M}\widetilde{M})^\bigtriangleup\neq 0$,}\\
$3,$ &\mbox{if $M\widetilde{\widehat{M}}(\widehat{M}\widetilde{M})^\bigtriangleup=0$ and} \\
&\mbox{$M\widetilde{\widehat{M}}\widehat{M}+M\widetilde{\widehat{M}}\widetilde{M}+M(\widehat{M}\widetilde{M})^\bigtriangleup+\widetilde{\widehat{M}}(\widehat{M}\widetilde{M})^\bigtriangleup\neq 0$,}\\
$2,$ &\mbox{if $M\widetilde{\widehat{M}}(\widehat{M}\widetilde{M})^\bigtriangleup=M\widetilde{\widehat{M}}\widehat{M}+M\widetilde{\widehat{M}}\widetilde{M}+M(\widehat{M}\widetilde{M})^\bigtriangleup+\widetilde{\widehat{M}}(\widehat{M}\widetilde{M})^\bigtriangleup=0$} \\ &\mbox{and $M\widetilde{\widehat{M}}+M\widehat{M}+M\widetilde{M}+\widetilde{\widehat{M}}\widehat{M}+\widetilde{\widehat{M}}\widetilde{M}+(\widehat{M}\widetilde{M})^\bigtriangleup\neq 0$,}\\
$1,$  &\mbox{if $M\widetilde{\widehat{M}}(\widehat{M}\widetilde{M})^\bigtriangleup=M\widetilde{\widehat{M}}\widehat{M}+M\widetilde{\widehat{M}}\widetilde{M}+M(\widehat{M}\widetilde{M})^\bigtriangleup+\widetilde{\widehat{M}}(\widehat{M}\widetilde{M})^\bigtriangleup=$}\\
&\mbox{$=M\widetilde{\widehat{M}}+M\widehat{M}+M\widetilde{M}+\widetilde{\widehat{M}}\widehat{M}+\widetilde{\widehat{M}}\widetilde{M}+(\widehat{M}\widetilde{M})^\bigtriangleup=0$ and $M\neq 0$,}\\
$0,$  &\mbox{if $M=0$.}
\end{cases}\label{exp22}
\end{eqnarray}
\end{example}

Note that formulas (\ref{rankherm}), (\ref{exp12}), (\ref{exp22}) differ from (\ref{rank22}), (\ref{exp1}), (\ref{exp2}), respectively, by replacing the multivector $T=M^\dagger M$ with the multivector $M$.

\section{Conclusions}\label{secConcl}

In this paper, we establish the notion of rank of multivector in complexified Clifford geometric algebras without using the corresponding matrix representations. Theorem \ref{thrank} involves only operations in geometric algebras. To obtain the results, we use the fact that rank of the matrix $A\in\Mat(N, \BC)$ is the number of nonzero (positive) singular values of $A$, that is, the number of nonzero (positive) eigenvalues of the matrix $A^\dagger A$. Theorems \ref{thrankpr}, \ref{thrankpr2}, \ref{thrank}, \ref{thrankherm} are new. New explicit formulas (\ref{exp1}), (\ref{exp2}) for the cases of dimensions $3$ and $4$ can be used in various applications of geometric algebras in physics, engineering, and computer science.

Note that the results of this work are valid not only for complexified Clifford geometric algebras, but also for real Clifford geometric algebras, since we can use the same matrix representations in the real case (but these matrix representations will have non-minimal dimension in this case).

Studying the relationship between rank of multivector and other classical matrix concepts related to rank, such as rows and columns, minors, and row echelon form in terms of geometric algebra operations without using the corresponding matrix representation is an interesting task for further research.

\begin{ack}[Acknowledgements]
The author is grateful to the anonymous reviewers for their careful reading of the paper and helpful comments on how to improve the presentation.
\end{ack}

\begin{ack}[Conflict of interest]
This work does not have any conflicts of interest.
\end{ack}

\bigskip 

Data sharing not applicable to this article as no datasets were generated or analysed during the current study.

\end{document}